\documentclass[a4,11pt]{article}
\usepackage{fullpage}

\usepackage{graphicx}
\usepackage{amsmath, amssymb, amsthm}
\usepackage{tikz}
\usepackage{todonotes}
\usepackage{nicefrac}

\newcommand{\cost}{\mathrm{cost}}

\newtheorem{theorem}{Theorem}[section]
\newtheorem{lemma}[theorem]{Lemma}
\newtheorem{corollary}[theorem]{Corollary}
\newtheorem{definition}[theorem]{Definition}
\newtheorem{fact}[theorem]{Fact}
\newtheorem{conjecture}{Conjecture}

\title{Sparse Relaxed Broadcast Graphs\thanks{This work was done while the first author was visiting the Department of  Computer Science and Software Engineering at Concordia University, Montréal.}}

\author{Pierre Fraigniaud\thanks{Institut de Recherche en Informatique Fondamentale (IRIF), CNRS and Université Paris Cité, France. Additional supports from ANR Projects ENEDISC (ANR-24-CE48-7768-01), and PREDICTIONS (ANR-23-CE48-0010), and from the InIDEX Project METALG.} 
\and 
Hovhannes Harutyunyan\thanks{Department of Computer Science \& Software Engineering, Concordia University, Montréal, Canada}}

\date{}

\begin{document}

\maketitle

\begin{abstract}
Broadcasting in graphs refers to the information dissemination problem in which a source node has an atomic piece of information to be distributed to all the nodes of a graph. In the standard \emph{telephone model}, broadcasting proceeds as a sequence of synchronous rounds, where, at each round, every informed node can transfer the information to at most one of its neighbors. The \emph{broadcast time} of a graph $G$ is the maximum, taken over every node~$v\in V(G)$, of the minimum number of rounds required for broadcasting from $v$ in~$G$. 
Since the number of informed nodes can at most double at each round, the broadcast time of any $n$-node graph is at least $\lceil\log_2 n\rceil$. We study the network design problem that, for every $\epsilon> 0$, asks for the minimum number of edges of $n$-node graphs with broadcast time close to optimal, i.e., at most $(1+\epsilon)\log_2n$. 

Let $\phi=(1+\sqrt{5})/2$ be the golden ratio, and let $\alpha=1/\log_2\phi-1\simeq 0.44$. The aforementioned problem is solved for $\epsilon\geq \alpha$ as Labahn (1989), and  Khachatrian and Haroutunian (1989) independently proved that, for every $n\geq 1$, there are $n$-node trees with broadcast time at most $\log_2n/\log_2\phi$. We address the problem for $\epsilon < \alpha$. We show that, for every~$n\geq 1$, and for every $\epsilon\in(0,\alpha)$, it suffices to add $O(n^{1-\epsilon/\alpha})$ edges to a well chosen $n$-node tree for designing an $n$-node graph with broadcast time $(1+\epsilon)\log_2n$. This asymptotic bound on the additional number of edges improves the previsouly known bound $O(n^{1-\epsilon})$ by Averbuch,  Peeri, and Roditty (2017), and has implications to the design of graphs with minimum broadcast \emph{cost}, defined as number of edges times broadcast time.  

Moreover, we show that the upper bound $2n-\lceil\log_2n\rceil-2$ by Grigni and Peleg (1991) on the minimum number of edges of an $n$-node graph with broadcast time $\lceil\log_2 n\rceil+1$ is of the correct order of magnitude in the sense that, for infinitely many values of~$n$, $\Omega(n)$ edges must be added to some tree for designing an $n$-node graph with broadcast time $\lceil\log_2 n\rceil+1$. Specifically, we show that, for every~$k$, all graphs with $n=2^k$ nodes and broadcast time $k+1$ must have at least $\frac98n$ edges. Therefore, our bound $O(n^{1-\epsilon/\alpha})$ on the additional number of edges for $0<\epsilon<\alpha$ is asymptotically tight at the two extremities of the interval $(0,\alpha]$, as it is $O(n)$ when $\epsilon\to 0$, and $O(1)$ when $\epsilon=\alpha$. 

Finally, we (slightly) improve the aforementioned upper bound $2n-\lceil\log_2n\rceil-2$ by Grigni and Peleg (1991) by showing that, for every~$n$, there exists an $n$-node graph with broadcast time $\lceil\log_2 n\rceil+1$ and at most $2n-4\lceil\log_2n\rceil+O(1)$ edges. 
\end{abstract}

\newpage

\section{Introduction}
\label{sec:introduction}

This paper is studying the construction of sparse networks (i.e., networks  with few communication links) still capable to support efficient communication primitives such as routing, or one-to-all and all-to-all primitives. The paper specifically focuses on a one-to-all communication primitive commonly  referred to as \emph{broadcasting}, in which a source node has to transmit a message to all the other nodes of the network. Our goal is, for any given number of nodes in the network, to identify the tradeoff between the minimum number of communication links of the network and the broadcast time of that network. 

In other words, the question addressed in this paper is: given a number $n$ of nodes, and a time bound~$t$, what is the minimum number of edges $m=m(n,t)$ such that there exists an $n$-node network with $m$ edges and broadcast time~$t$? We address this question in the standard \emph{telephone model}~\cite{FarleyHMP79,HedetniemiHL88}.

\subsection{Context and Objective} 

Let $G=(V,E)$ be a graph, let $s\in V$, and let us assume that $s$ is given an atomic piece of information to be distributed to all the nodes of~$G$. A broadcast protocol from $s$ in $G$ under the telephone model proceeds as  a sequence of synchronous rounds. At each round, every node that is aware of the information broadcasted from~$s$ can transmit this information to at most one of its neighbors in~$G$ during the round. The minimum number of rounds required to broadcast the information from $s$ to all nodes of $G$ is called the \emph{broadcast time of $s$ in~$G$}, and is denoted by $b(G,s)$. We let $b(G)=\max_{s\in V}b(G,s)$ denote the \emph{broadcast time of~$G$}. See Figure~\ref{fig:mbg12} for an illustration. 

\begin{figure}[htb]
\centerline{\includegraphics[scale=.5]{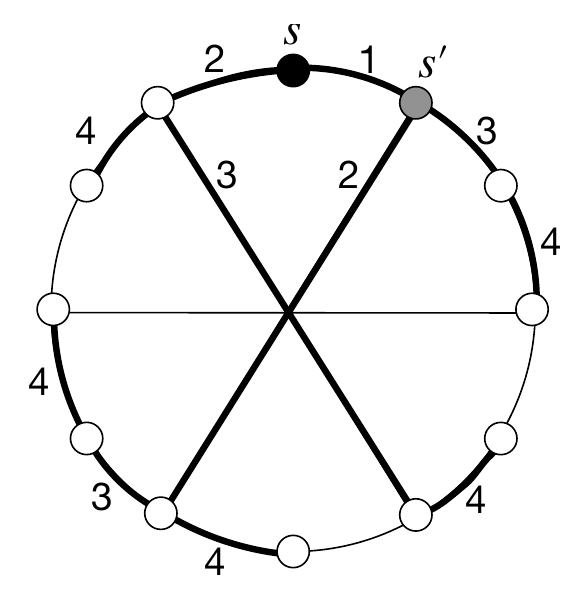}}
\caption{\sl A 12-node graph $G$, with broadcast time $b(G)=4=\lceil\log_212\rceil$, taken from~\cite{FarleyHMP79}. A broadcast protocol from node~$s$, as well as the one from node~$s'$, is illustrated in bold, where the numbers indicate the rounds at which the calls take place. This graph has 15 edges, and no 12-node graphs with less than 15 edges can achieve broadcast time at most~4 (see \cite{FarleyHMP79}), hence the minimum number of edges of 12-node graphs guaranteeing broadcast in at most 4 rounds is~$15$. }
\label{fig:mbg12}
\end{figure}

Since an informed node can inform at most one other node at each round,  the number of nodes aware of the information can at most double at each round, and thus at most $2^r$ nodes may be aware of the information after $r$ rounds.  As a consequence 
\[
b(G,s)\geq \lceil \log_2 n\rceil
\]
for every $n$-node graph $G=(V,E)$, and any node $s\in V$.  The \emph{broadcast problem} consists, given a graph $G=(V,E)$, and a node  $s\in V$, to compute the broadcast time $b(G,s)$. There is a vast literature on the broadcast problem, which was mainly approached according to two different but complementary directions summarized hereafter.

\subsubsection{Minimizing the Broadcast Time} 

The decision version of the broadcast problem, i.e., deciding whether $b(G,s)\leq k$ for any given graph~$G=(V,E)$, any given node $s\in V$, and any given integer~$k$, is NP-complete~\cite{SlaterCH81}. It remains NP-hard even for graphs with feedback vertex number (the cardinality of a minimum feedback vertex set) at most~1, and hence treewidth at most~2~\cite{Tale25}. It also remains NP-hard in cactus graphs, and in graphs of pathwidth at most~2~\cite{AminianKJS25}. Actually, the problem is hard even in graphs of bounded treedepth~\cite{EgamiGHKLMNOV025}. On the other hand, there exists a linear-time algorithm for the broadcast problem in trees~\cite{SlaterCH81}. A series of papers~\cite{Bar-NoyGNS98,KortsarzP95,Ravi94} eventually yield a polynomial-time $O(\frac{\log n}{\log\log n})$-approximation algorithm~\cite{ElkinK03}.  Randomized broadcast protocols --- at each round, every informed node transmits the information to a neighbor chosen uniformly at random --- were shown to be efficient in large classes of graphs (see~\cite{FeigePRU90}). For instance, in the Erdős–Rényi random graph model~$\mathcal{G}_{n,p}$, it holds that, for almost all graphs $G\in \mathcal{G}_{n,p}$,  the randomized broadcast protocol completes in $O(\log n)$ rounds w.h.p. More recently, the problem was tackled in the context of parametrized algorithms. It was shown that the broadcast problem is FPT when parametrized by the size of a feedback edge set, and when parametrized by the size of a vertex cover~\cite{BonnetFV2026,FominFG23}, as well as  when parametrized by vertex integrity~\cite{EgamiGHKLMNOV025}.  

\subsubsection{Minimizing the Number of Edges} 

A dual line of research addresses broadcasting as a \emph{network design} issue, and aims at constructing the sparsest graphs with given order~$n$, and given broadcast time. For every integer $n\geq 1$, let us denote by $B(n)$ the minimum number of edges of any $n$-node graph $G$ with optimal broadcast time, i.e., such that $b(G)= \lceil \log_2 n\rceil$ (see~\cite{FarleyHMP79}). Note that the function $B:\mathbb{N}\to \mathbb{N}$ is well defined as, for every $n\geq 1$, $b(K_n)=\lceil \log_2 n\rceil$, where $K_n$ denotes the complete graph on $n$ vertices. Every $n$-node graph $G$ satisfying
\[
b(G)=\lceil \log_2 n\rceil
\;\mbox{and}\; |E(G)|=B(n)
\]
is called a \emph{minimum broadcast graph}, or MBG for short. The graph on Figure~\ref{fig:mbg12} is a 12-node MBG, i.e., $B(12)=15$.
It was shown~\cite{GrigniP91} that 
\[
B(n)=\Theta(n\cdot L(n)),
\]
where $L(n)$ denotes the number of leading 1s in the binary representation of $n-1$. As a consequence, the sparsest graphs that allow broadcasting in $\lceil \log_2 n\rceil$ rounds may have up to $\Theta(n\log n)$ edges for infinitely many~$n$. For example, the hypercube is known to be a minimum broadcast graph, as every graph with $n=2^k$ nodes and broadcast time~$k$ must have minimum degree at least~$k$, and thus at least $k\, 2^{k-1}$ edges. 

The notion of \emph{relaxed} minimum broadcast graphs (or relaxed MBGs) was also introduced in~\cite{GrigniP91}, and recently revisited and extended in~\cite{AverbuchPR17}. 

\begin{definition}
    For every integers $n\geq 1$ and $\tau \geq 0$, let $B(n,\tau)$ denote the minimum number of edges of any $n$-node graph $G$ with broadcast time at most $\lceil \log_2 n\rceil+\tau$. The parameter $\tau$ is called \emph{relaxation}. An $n$-node graph with $B(n,\tau)$ edges and broadcast time $b(G)\leq \lceil \log_2 n\rceil+\tau$ is called a \emph{$\tau$-relaxed minimum broadcast graph}, or \emph{$\tau$-relaxed MBG} for short.
\end{definition}

Note that, in particular, $B(n,0)=B(n)$. Note also that the relaxation parameter $\tau$ may  be a function of~$n$, e.g., $\tau=\epsilon\cdot\log_2n$ for studying graphs with broadcast time $(1+\epsilon)\log_2n$. 
 
\paragraph{Remark.} 

It is in fact more informative to express the number of edges $B(n,\tau)$ as 
\[
B(n,\tau)=(n-1)+h(n,\tau).
\]
Indeed, $h(\cdot)$ is then the \emph{overhead} to be paid in addition to the unavoidable connectivity requirement imposing $n-1$ edges, i.e., the number of edges of any tree spanning the graph. The network designer has to pay a fixed price  anyway for setting up $n-1$ links interconnecting the $n$ users, and its goal is to minimize the overhead cost induced by adding links for improving services, e.g.,  speeding up communications among these users. The goal of the network designer is to identify the tradeoff between this overhead cost and the resulting improvement in the quality of service.

\medskip

The state of the art on the matter is the following:

\begin{itemize}
\item For a smallest relaxation $\tau=1$, it is known~\cite{GrigniP91} that, for every $n\geq 1$, there is a graph $G$ with $2n-\lceil \log_2 n\rceil-2$ edges and broadcast time $b(G)= \lceil \log_2 n\rceil+1$. Therefore, the overhead on the number of edges of minimum 1-relaxed broadcast graphs satisfies $h(n,1)\leq n-\lceil \log_2 n\rceil-1$. That is, the overhead drops abruptly from $h(n,0)=\Theta(n\cdot L(n))$ to $h(n,1)=O(n)$ by relaxing the broadcast time from $\lceil \log_2 n\rceil$ to $\lceil \log_2 n\rceil+1$.

\item Another threshold of interest occurs for relaxation $\tau=\alpha\cdot\log_2 n$, with 
\[
\alpha=1/\log_2\phi-1
\]
where $\phi=(\sqrt{5}+1)/2$ is the golden ratio (i.e., $\alpha\simeq 0.44$). Indeed, it was shown~\cite{KhachatrianH89,Labahn89} that, for every $n\geq 1$, there exists an $n$-node  \emph{tree} with broadcast time at most $\log_2n/\log_2\phi$. In other words, the overhead $h(n,\tau)$ is null when the relaxation satisfies $\tau\geq \alpha\cdot\log_2 n$. 

\item For the intermediate cases, i.e., for any relaxation $\tau=\epsilon\cdot\log_2 n$ with $0<\epsilon< \alpha$, it was shown~\cite{AverbuchPR17} that the overhead satisfies
\[
h(n,\epsilon\cdot\log_2 n)=O(n^{1-\epsilon}).
\]
\end{itemize}

\bigbreak

\noindent All these results are summarized in Figure~\ref{fig:summary}. In particular, this figure illustrates that the upper bound $O(n^{1-\epsilon})$ on the the overhead established in~\cite{AverbuchPR17}  decays smoothly when the relaxation $\tau=\epsilon\cdot \log_2n$ grows with~$\epsilon$. However, this decay is insufficiently  fast to reach~0 for $\tau=\alpha\log_2 n$, whereas the overhead is null for any relaxation $\tau\geq \alpha\log_2 n$~\cite{KhachatrianH89,Labahn89}. Instead, for a relaxation~$\tau=\alpha\log_2 n$, the upper bound on the overhead established in~\cite{AverbuchPR17} is $O(n^{1-\alpha})$, which is roughly $O(n^{0.56})$.  

\begin{figure}[htb]
\centerline{\includegraphics[scale=0.4]{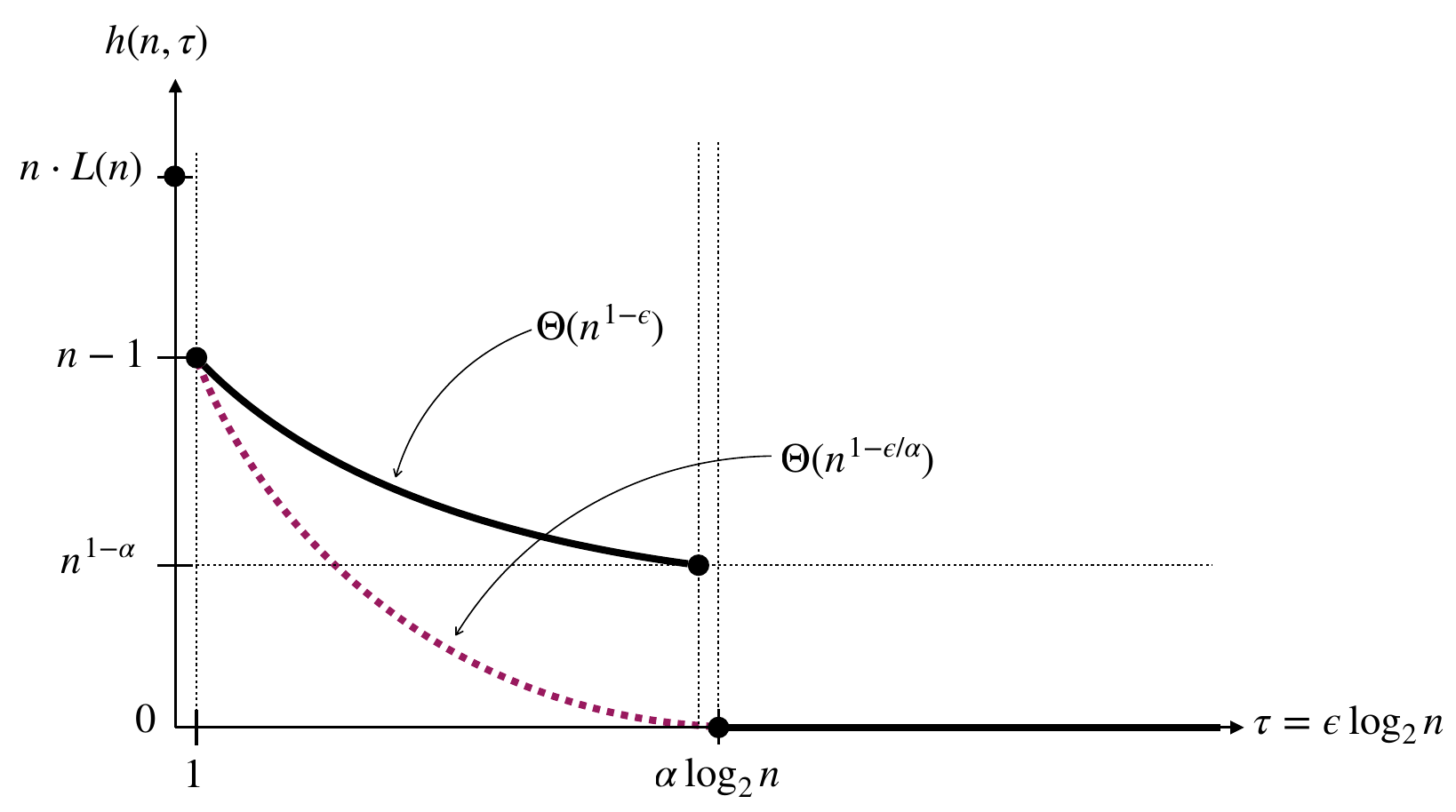}}
\caption{\sl Best known upper bounds on the overhead $h(n,\tau)$ as a function of the relaxation~$\tau=\epsilon\log_2n$, when the number of nodes $n$ is fixed (and large). The plain bold points and lines represent the state of the art before this paper. The dotted line is the new upper bound on the overhead provided in this paper. The parameter  $\alpha\simeq 0.44$ is defined as $\alpha=1/\log_2\phi-1$ where $\phi=(\sqrt{5}+1)/2$ is the golden ratio.  }
\label{fig:summary}
\end{figure}

Note that there is a huge gap $\Theta(n\cdot L(n))$ in the overhead values occurring between relaxations $\tau=0$ and $\tau=1$, so it is a priori unclear whether another gap at $\tau=\alpha\log_2 n$ should occur or not. In this paper, we question whether the gap of $\Theta(n^{1-\alpha})$ for the overhead occurring when the relaxation value $\tau$ reaches $\alpha \log_2 n$  effectively  occurs. 

\subsection{Our Results}

Our results are summarized in Table~\ref{tab:summary}, in comparison to the state-of-the-art. In a nutshell, we first establish a new upper bound $O(n^{1-\epsilon/\alpha})$ on the overhead $h(n,\tau)$ for any relaxation $\tau=\epsilon\log_2n$ with $0<\epsilon<\alpha$, which improves the bound in~\cite{AverbuchPR17}. Next, we show a lower bound $\frac18n$ on the overhead for $\tau=1$ and $n=2^k$, which shows that the linear overhead $n-\log_2n+O(1)$ in~\cite{GrigniP91} is of the right order of magnitude. Finally, we (slightly) improve the overhead in~\cite{GrigniP91} from $n-\log_2n+O(1)$ to $n-4\log_2 n+O(1)$, showing that while the overhead for $\tau=1$ is linear, it is not clear whether it is of the form $n-o(n)$, or some fraction of~$n$. Let us now enter into the details of our results. 

\begin{table}[tb]
\begin{center}
\begin{tabular}{|c|c|c|c|}
\hline
& overhead $h(n,\tau)$ & overhead $h(n,\tau)$  &  \\
relaxation $\tau$ & upper bound & lower bound & reference \\
\hline
0 & $O(n\cdot L(n))$ & $\Omega(n\cdot L(n))$ & \cite{GrigniP91} \\
1 & $n-\log_2 n+O(1)$ & - & \cite{GrigniP91} \\
$\epsilon\log_2n$ with $0<\epsilon<\alpha$ & $O(n^{1-\epsilon})$ & - & \cite{AverbuchPR17} \\
$\epsilon\log_2n$ with $\epsilon\geq \alpha$ & $0$ & - & \cite{KhachatrianH89,Labahn89} \\
$\epsilon\log_2n$ with $0<\epsilon<\alpha$ & $O(n^{1-\epsilon/\alpha})$ & - & [this paper] \\
1 & $n-4\log_2n+O(1)$ & $\frac18n$& [this paper] \\
\hline
\end{tabular}
\end{center}
\caption{\sl State of the art, and our new results, with $\alpha=1/\log_2\phi-1$ where $\phi=(\sqrt{5}+1)/2$ is the golden ratio. The lower bound $\frac18n$ holds for powers of~2. }
\label{tab:summary}
\end{table}

\subsubsection{Upper bound} 

Our upper bound $O(n^{1-\epsilon/\alpha})$ on the overhead $h(n,\tau)$ for a relaxation $\tau=\epsilon\cdot\log_2n$ is a direct consequence of the following theorem established in this paper. 

\begin{theorem}
\label{thm:upper-bound-general}
Let $\alpha=1/\log_2\phi-1$ where $\phi=(\sqrt{5}+1)/2$ is the golden ratio, and let $\epsilon\in (0,\alpha)$. For every $n\geq 1$, there is an $n$-node graph with  broadcast time $(1+\epsilon)\log_2n$, and $n+O(n^{1-\epsilon/\alpha})$ edges.
\end{theorem}

Figure~\ref{fig:summary}  illustrates this result, by displaying the dotted curve $h(n,\tau)=n^{1-\epsilon/\alpha}$ for ${0<\epsilon<\alpha}$. 
Our bound is asymptotically tight when $\epsilon=\alpha$. 

\paragraph{Our Techniques.}

Our construction is different from the construction in~\cite{AverbuchPR17}. The latter essentially constructs graphs by connecting hypercubes of different dimensions together in a sophisticated manner. Instead, our construction is roughly as follows (details will be provided later). It is merely based on a ``core graph''~$C$ in which each node $v$ is also the root of a tree $T_v$ of depth~$m$ (See Figure~\ref{fig:construction}). All the trees are isomorphic to a same rooted tree~$T$. Hence, the total number of nodes is 
\[
n=|V(C)|\cdot |V(T)|.
\]
The tree $T$ is obtained by ``truncating'' at level~$m$ the Binomial tree of dimension~$k$ for suitable parameters $m$ and~$k$. Recall that Binomial trees form a family of rooted trees defined as follows. The binomial tree $B_0$ consists of a single vertex. For $k\geq 1$, the $k$-dimensional binomial tree $B_k$ consists of two binomial trees $B_{k-1}$ whose roots are connected by an edge, and where one of these two roots is selected as the root of $B_k$. The number of nodes at level $i\in\{0,\dots,k\}$ in $B_k$ is $\binom{k}{i}$, hence the name of these trees. The tree $T$ is obtained from $B_k$ by keeping only the nodes at distance at most $m$ from the root, i.e., by keeping only the levels $0,\dots,m$ of $B_k$ (See Figure~\ref{fig:construction}). The following summarizes the main properties of the (possibly truncated) Binomial trees used throughout the paper.

\begin{fact}\label{basic-fact}
    For every $k\geq 0$, the number of nodes in the $k$-dimensional Binomial tree $B_k$ is~$2^k$, and if $r$ denotes the root of $B_k$, then $b(B_k,r)=k$. The tree $T$ resulting from  truncating the Binomial tree $B_k$ at level $m\in\{1,\dots,k\}$, i.e., from keeping only nodes at distance at most $m$ from the root~$r$ of~$B_k$, has $\sum_{i=0}^m\binom{k}{i}$ nodes, depth~$m$, and satisfies $b(T,r)=k$. Conversely, any rooted tree $T$ satisfying $b(T,r)\leq k$ for its root $r$ is a subtree of a Binomial tree $B_k$ rooted at~$r$. 
\end{fact}

Let us denote by $d$ be the broadcast time of the core graph~$C$, i.e., $d=b(C)$. The graph $G$ resulting from our construction has broadcast time 
\[
b(G)\leq m+d+k.
\]
Indeed, (1)~any node is at distance at most $m$ from a node $v$ in the core (by Fact~\ref{basic-fact}), (2)~broadcasting from $v$ to all nodes in $C$ takes at most $d$ rounds (by the definition of~$d$), and (3)~broadcasting from any node $u\in V(C)$ in the tree $T_u$ takes at most $k$ rounds (by Fact~\ref{basic-fact} again). 

Given a relaxation value $\tau=\epsilon\cdot\log_2n$, our strategy thus consists to determine what is the appropriate core graph~$C$, and what are the appropriate parameters $k,m$, and $d$ so that, first, $b(C)=d$, second, 
\[
m+d+k\leq \lceil\log n\rceil+\tau, 
\]
and, last but not least, the number of edges $|E(G)|$ of~$G$ is minimized, where
\[
|E(G)|=|E(C)|+|V(C)| \cdot |E(T)|=|E(C)|+|V(C)| \cdot\Big(\sum_{i=0}^m\binom{k}{i}-1\Big).
\]
Note that choosing $C$ as the $d$-dimensional hypercube~$Q_d$ is one option, as the broadcast time of~$Q_d$ is~$d$, as desired. However, $Q_d$ has $d\, 2^{d-1}$ edges. So, as we shall show, picking a sparser graph as the core graph~$C$ is actually a much better option, even if this choice comes with the cost of decreasing the number of nodes in the core for preserving a broadcast time~$d$. 

\begin{figure}[tb]
\centerline{\includegraphics[scale=0.45]{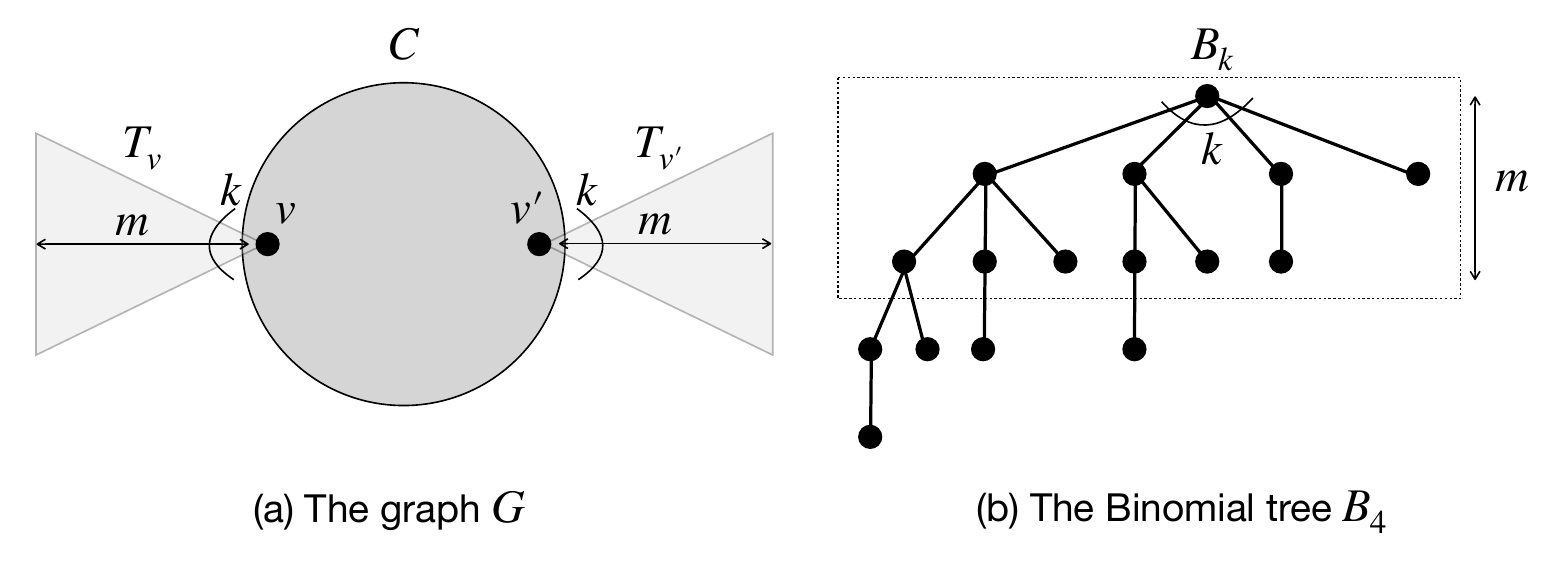}}
\caption{\sl  Illustration of our construction. (a)~Our graph $G$ consists of a core graph $C$ to which is attached a tree $T_v$ at each node~$v$. Each tree $T_v$ is isomorphic to a ``truncated'' Binomial tree $B_k$ at level~$m$. (b)~A k-dimensional Binomial tree $B_k$ for $k=4$. The doted box displays $B_4$ truncated at level $m=2$ (the root is at level~0). }
\label{fig:construction}
\end{figure}

\paragraph{Remark.}

As is, the above construction produces graphs with a number of nodes of a specific form $n=|V(C)|\cdot |V(T)|$. To get graphs with arbitrary number of vertices, we shall  remove nodes from some of the trees~$T_v$, $v\in V(C)$, in order to get a graph $G$ with the desired total number of vertices. For instance, the tree obtained by truncating $B_4$ at level~2 has $11$ nodes (see Figure~\ref{fig:construction}). One can ``prune''  this tree further, by removing as many nodes as desired, arbitrarily, starting from the leaves, while preserving a broadcast time at most 4 from the root.

\paragraph{Application to the broadcast cost.} 

An interesting consequence of Theorem~\ref{thm:upper-bound-general} is that it directly follows from it that the minimum value of the \emph{broadcast cost} of any $n$-node graph $G$, defined in~\cite{Harutyunyan14} as $\cost(G)=b(G)\cdot|E(G)|$ can be upper bounded. Specifically, for every $n\geq 1$, there exists an $n$-node graph $G$ with broadcast cost $n\log_2n+O(n\log\log n)$, which matches the upper bound in~\cite{Harutyunyan14}, and is close to the lower bound $n\log_2n+\Omega(n)$ from~\cite{Harutyunyan14}.

\subsubsection{Lower bound}

Deriving lower bounds on the minimum number of edges of any $n$-node graph with prescribed broadcast time  is notoriously difficult, as witnessed by the fact that exact bounds on the value of the minimum number of edges $B(n)$ of any $n$-node graph $G$ with (optimal) broadcast time $b(G)= \lceil \log_2 n\rceil$ are only known for very few ranges of~$n$. For instance, it is not even known whether the function $B$ is non-decreasing between two consecutive powers of~2, i.e., wether $B(n+1)\geq B(n)$ for all $n\in[2^k,2^{k+1}-1]$ and $k\geq 0$ (see~\cite{HarutyunyanL03}). This result is known to hold only for the first quarter of the interval $[2^k,2^{k+1}-1]$ \cite{HarutyunyanL12}. 

The problem is even more difficult for $\tau$-relaxed broadcast time $b(G)= \lceil \log_2 n\rceil+\tau$. One reason for this is that the minimum degree of the sparsest graphs with broadcast time $b(G)= \lceil \log_2 n\rceil+1$ may  be constant, which ruins one of the most fruitful approach for designing lower bounds on the number of edges, which consists of arguing about the minimum degree of the nodes. Nevertheless, we were able to come up with a linear lower bound on $B(n,1)$, i.e., on the minimum number  of edges of  $n$-node graphs with broadcast time $\lceil \log_2 n\rceil+1$.  

\begin{theorem}
\label{thm:lower-bound}
For every $k\geq 1$ and $n=2^k$, any $n$-node graph with broadcast time $k+1$ has at least $\frac98n-O(1)$ edges.
\end{theorem}

A direct consequence of this theorem is that the overhead $h(n,1)$ for 1-relaxed minimum broadcast graphs satisfies $h(n,1)\geq \frac{n}{8}-O(1)$, i.e., $h(n,1)$ is linear in~$n$, for an infinite number of integers~$n$. This shows that the best known upper bound $h(n,1)\leq n-\log_2n +O(1)$~\cite{GrigniP91} has the right order of magnitude, that is $h(n,1)=O(n)$, and this is tight for infinitely many~$n$. In particular, our general upper bound $O(n^{1-\epsilon/\alpha})$ is tight when $\epsilon\to 0$.  

\subsubsection{1-Relaxed Minimum Broadcast Graphs}

Finally, we show that the upper bound $n-\lceil\log_2n\rceil-2$ on the overhead $h(n,1)$ in~\cite{GrigniP91} is not tight. This is the consequence of the following result, which shows that $h(n,1)\leq n-4\log_2n+O(1)$. 

\begin{theorem}
\label{thm:upper-bound-tau-1}
    For every $n\geq 1$, there exists an $n$-node graph with broadcast time $\lceil\log_2n\rceil+1$ and at most $2n-4\log_2n+O(1)$ edges. 
\end{theorem}

The improvement over~\cite{GrigniP91}, from $n-\lceil\log_2 n\rceil+O(1)$ to $n-4\lceil\log_2 n\rceil+O(1)$, is rather small. Nevertheless, it is worth mentioning two points. First, Theorem~\ref{thm:upper-bound-tau-1} demonstrates that $n-\lceil\log_2 n\rceil+O(1)$ is not the right answer, and, in view of the lower bound~$\frac{n}{8}$ established in Theorem~\ref{thm:lower-bound}, this is an invitation for a more significant improvement.  Second, the graph resulting from Theorem~\ref{thm:upper-bound-tau-1} can be used as the core graph in the proof of Theorem~\ref{thm:upper-bound-general}. 

\subsection{Related Work}

A detailed survey of the early work in broadcasting and gossiping (i.e., simultaneous broadcasting from all the nodes) can be found in~\cite{HedetniemiHL88}. Broadcasting in specific classes of graphs has been studied extensively (see the surveys \cite{FraigniaudL94,HarutyunyanLPR13,Hromkovic05}), including random graphs~\cite{FriezeM94}. Similarly,  the design of minimum broadcast graphs has been studied extensively (cf., e.g., \cite{AverbuchSR14,BermondFP95,HarutyunyanL19,HarutyunyanL20,KhachatrianH90,VenturaW93,WengV94}). In particular, it was shown~\cite{GrigniP91} that $B(n,0)=\Theta(n\cdot L(n))$ and  $B(n,1)=\Theta(n)$, where $L(n)$ denotes the number of leading 1s in the binary representation of $n-1$. The significant drop in the minimum number of edges between $n$-node graphs broadcasting in $\lceil\log_2n\rceil$ rounds, and graphs broadcasting in $\lceil\log_2n\rceil+1$ rounds motivated the study of \emph{time-relaxed} minimum broadcast graphs, initiated in~\cite{Shastri98} --- although the study of graphs broadcasting in $\lceil\log_2n\rceil+\tau$ rounds was also considered in~\cite{Liestman85}, but in the context of fault-tolerance (see also \cite{AGHK96,RiekstsV09}).  

General techniques for constructing sparse graphs with broadcast time $\lceil\log_2n\rceil+\tau$ were presented in~\cite{DinneenVWZ99}. The first generic upper bound on the minimum number  of edges, $B(n,\tau)$, of $\tau$-relaxed minimum broadcast graphs was given in~\cite{AverbuchPR17}. For $\tau=\epsilon\log n$ with $0<\epsilon<1$,  this bound is $B(n,\tau)=n+O(n^{1-\epsilon})$. This bound is however not tight, at least for $\epsilon \geq \alpha$. Indeed, as mentioned before, it is known~\cite{KhachatrianH89,Labahn89} that, for every $n\geq 1$, there exists an $n$-node tree with broadcast time at most $(1+\alpha)\log_2n$.  Note that $B(n,k)=n-1$ trivially holds for $k=\lceil\log_2n\rceil$. This bound is merely obtained using the Binomial tree $B_k$ as any node of $B_k$ is at distance at most $k$ from the root of~$B_k$, and broadcasting from the root of $B_k$ takes $k$ rounds. 

\paragraph{Remark.} 

Several of the papers cited above consider a general variant of the telephone model in which every informed node can call $q\geq 1$ nodes at each round, instead of just one as in the standard model. In this general case, the minimum broadcast time of $n$-node graphs is $\lceil\log_{q+1}n\rceil$, and $\tau$-relaxed $q$-broadcast graphs are graphs broadcasting in at most $\lceil\log_{q+1}n\rceil+\tau$ rounds. It was shown that the results  for $q=1$ in there extend smoothly to arbitrary fixed integer~$q\geq 1$, without occurrence of any disruption effect. The same holds for the results in this paper. 

\medskip

Finally, it is worth mentioning that, due to its relevance to several practical applications in parallel and distributed computing, as well as in networking in general, the broadcast problem has been considered in a wide range of models, including the \emph{affine} model~\cite{JohnssonH89}, the \emph{postal} model~\cite{Bar-NoyK97}, and several variants of the (randomized) \emph{push-pull} model (see, e.g., \cite{AvinE18,BerenbrinkCEG10,KarpSSV00}) to mention just a few. In several of these papers, not only ``one-to-all'' communication primitives  such as broadcasting are considered, but also \emph{gossiping} (all-to-all communication), \emph{scattering} (personalized one-to-all communication), and \emph{multi-scattering} (personalized all-to-all communication).

\subsection{Organization of the Paper}

The next section (Section~\ref{sec:main-construction}) provides the proof of Theorem~\ref{thm:upper-bound-general}. Sections~\ref{sec:lwb} and~\ref{sec:1-relaxed} respectively present the proofs of Theorems~\ref{thm:lower-bound} and~\ref{thm:upper-bound-tau-1}. Finally, Section~\ref{sec:conclusion} lists concluding remarks, and conjectures. 

\section{Construction of Sparse Relaxed MBGs}
\label{sec:main-construction}

In this section, we prove our general upper bound $O(n^{1-\epsilon/\alpha})$ on the overhead for $(\epsilon\log_2n)$-relaxed minimum broadcast graphs, as stated in Theorem~\ref{thm:upper-bound-general}. That is, we show that, for every $\epsilon\in (0,\alpha)$, and every $n\geq 1$, there is an $n$-node graph $G$ with  broadcast time $(1+\epsilon)\log_2n$, and $n+O(n^{1-\epsilon/\alpha})$ edges. 

\begin{proof}[Proof of Theorem~\ref{thm:upper-bound-general}]
Let $\epsilon\in (0,\alpha)$, and let $n_\epsilon\geq 1$ be a value depending on $\epsilon$ that will be fixed later. We are constructing an $n$-node $G$ satisfying the statement of the theorem for every $n\geq n_\epsilon$. (For smaller~$n$, any broadcast graph on $n$ vertices works.) Let us first recall the general structure of the desired relaxed broadcast graph, as sketched in Section~\ref{sec:introduction} (see Figure~\ref{fig:construction}).

\paragraph{General Structure of our Relaxed Broadcast Graph.}

Our construction is based on a core graph~$C$, to which is attached a Binomial tree $B_k$ at every node, truncated at level~$m$ (see Figure~\ref{fig:construction}). The tree $T_v$ rooted at each node $v$ of the core $C$ has $\sum_{i=0}^{m}\binom{k}{i}$ nodes, and if $|V(C)|=\nu$ then the total number of nodes of the resulting graph $H$ is $\nu \sum_{i=0}^{m}\binom{k}{i}$. Let us denote by $d$ the broadcast time of the core, i.e., $b(C)=d$. Then, thanks to Fact~\ref{basic-fact}, 
\[
b(H)\leq m+d+k.
\]
Indeed, for every node $u$ of $H$, it takes at most $m$ round to reach the root $v$ of the Binomial tree $T_v$ containing~$u$, then $d$ rounds to broadcast from $v$ to all nodes of~$C$, and finally $k$ rounds for allowing each node of the core to broadcast in the Binomial tree rooted at that node. 
Let $\mu$ denotes the number of edges of the core~$C$. Our goal is then to solve
\[
\min_{m,k,C} \;\; \mu + \nu \cdot \Big(\sum_{i=0}^{m}\binom{k}{i}-1\Big)
\]
under the two constraints 
\[
(i)\;\; \nu \sum_{i=0}^{m}\binom{k}{i}\geq n,
\;\;\mbox{and}\; \; (ii)\;\; 
m+d+k\leq (1+\epsilon)\log_2n.
\]
The first constraint will help us to eventually obtain a graph $G$ with exactly $n$ nodes out of the $\nu \sum_{i=0}^{m}\binom{k}{i}$ nodes of the graph~$H$. The second constraint is to guarantee that broadcasting in $H$ does perform in the desired amount of rounds. These two constraints combine gently as we shall show that reducing the number of nodes from $\nu \sum_{i=0}^{m}\binom{k}{i}$ down to $n$ does not increase the broadcast time, i.e., $b(G)\leq b(H)$. 

\paragraph{Choice of the core graph $C$.} 

It is tempting to choose the core as the $d$-dimensional hypercube~$Q_d$, as this graph with $\nu=2^d$ nodes has broadcast time $b(Q_d)=d=\log_2 \nu$, which is optimal. However, the hypercube is too dense for our purpose. 
Instead, we choose the core $C$ as a 1-relaxed broadcast graph with $\nu$ nodes, at most $2\nu$~edges, and broadcast time $b(C)=\lceil\log_2n\rceil+1$ rounds, using our construction from Theorem~\ref{thm:upper-bound-tau-1}, or the previously known construction in~\cite{GrigniP91}. More specifically, we set 
\[
\nu=2^{d-1},
\]
such that $b(C)=d$. 
The core graph $C$ with $2^{d-1}$ nodes can merely be obtained from the $(d-1)$-dimensional Binomial tree $B_{d-1}$ by adding edges so that all nodes but the root of $B_{d-1}$ are directly connected to that root by an edge. It follows that the number of edges of the core is 
\[
\mu\leq 2\nu=2^d.
\]
Moreover, to broadcast in $d$ rounds, every source node can send its message to the root at round~1, and the $d-1$ remaining rounds are used to broadcast from be root in $B_{d-1}$ (cf. Fact~\ref{basic-fact}). 

\paragraph{Fixing the parameters.}

Once the general structure of the core is fixed, it remains to choose $d,k$, and~$m$. Let $t=m+d+k$ be the upper bound on the broadcast time of~$H$, i.e., we have $k=(t-d)-m$. For fixed parameters $t$ and~$d$, our goal is to maximize the value of  
\[
\sum_{i=0}^m\binom{(t-d)-m}{i}.
\]
Indeed, by doing so, we are maximizing the number of nodes in each of the $\nu=2^{d-1}$ trees attached to the $\nu$-node core. Each such node contributes for a single edge to the total number of edges (the edge connecting it to its parent in the tree), which is the best that can be expected from a node. 

For maximizing $\sum_{i=0}^m\binom{(t-d)-m}{i}$, we use the following technical lemma. 

\begin{lemma}[\cite{HarutyunyanL01,KhachatrianH89,Labahn89}]
\label{lem:technical-asymp-1}
    For every non-negative integer~$s$, and every $m\in\{0,\dots,s\}$, let $f(s,m)=\sum_{i=0}^m\binom{s-m}{i}$, and let $F(s)=\max_{0\leq m\leq s}f(s,m)$. Let $\phi=\nicefrac12(\sqrt{5}+1)$ be the golden ratio.
    When $s\to+\infty$, we have 
    $
    F(s)=(1\pm o(1))\cdot \phi^s.
    $
    Moreover, the maximum is reached for 
    $
    m=(1\pm o(1)) \beta s
    $
    where $\beta=\nicefrac12(\sqrt{5}-1)/\sqrt{5}\approx 0.28$.  
\end{lemma}

It follows from this lemma that, whenever $t-d$ grows to infinity, the number of nodes in the graph $H$ obtained from the core graph $C$ in which a  Binomial tree truncated at level $m=(1\pm o(1))\beta (t-d)$ is attached to each node satisfies 
\[
|V(H)|=2^{d-1}\cdot (1\pm o(1))\cdot \phi^{t-d}. 
\]
We can now choose $d$ such that 
\[
|V(H)|\geq n > 2^{d-1}
\]
by setting $t=(1+\epsilon)\log_2n$, and solving the double inequality 
\[
2^{d-1}(1\pm o(1))\; \phi^{(1+\epsilon)\log_2n-d} \geq n > 2^{d-1}.
\]
For the first inequality, by taking the $\log_2$ on both side, we get:
\begin{align*}
     (d-1)  &+ \log_2(1\pm o(1)) + \big((1+\epsilon)\log_2n-d\big) \log_2\phi \geq \log_2n\\
    \iff & \; d\;(1-\log_2\phi) \geq \big(1-(1+\epsilon)\log_2\phi\big)\log_2n +1-\log_2(1\pm o(1))\\
     \iff & \; d\geq \frac{1-(1+\epsilon)\log_2\phi}{1-\log_2\phi}\log_2n +\frac{1-\log_2(1\pm o(1))}{1-\log_2\phi}\\
    \iff & \; d\geq \frac{1-(1+\epsilon)\log_2\phi}{1-\log_2\phi}\log_2n +\frac{1\pm o(1)}{1-\log_2\phi}.
\end{align*}
As $\epsilon>0$, we have $\frac{1-(1+\epsilon)\log_2\phi}{1-\log_2\phi}<1$, and thus there exists $n_\epsilon$ for which 
\[
\frac{1-(1+\epsilon)\log_2\phi}{1-\log_2\phi}\log_2n_\epsilon +\frac{1\pm o(1)}{1-\log_2\phi}<\log_2n_\epsilon.
\]
For every $n\geq n_\epsilon$, one can thus set $d$ as the \emph{smallest} integer satisfying 
\begin{equation}\label{eq:interval-for-d}
    \frac{1-(1+\epsilon)\log_2\phi}{1-\log_2\phi}\log_2n +\frac{1\pm o(1)}{1-\log_2\phi}\leq d < \log_2n+1, 
\end{equation}
which yields $|V(H)|\geq n>2^{d-1}$, as desired. 

Our choice of $d$ as the smallest integer at least as large as the lower bound in Eq.~\eqref{eq:interval-for-d} guarantees that 
\[
d\leq \frac{1-(1+\epsilon)\log_2\phi}{1-\log_2\phi}\log_2n +O(1), 
\]
and thus 
\[
|E(H)|\leq n+|E(C)|\leq n+2^d = n + O\Big(n^{\frac{1-(1+\epsilon)\log_2\phi}{1-\log_2\phi}}\Big),  
\]
where the negligible $o(1)$-terms are now hidden in the big-O notation. 
Since $\frac{1}{\log_2\phi}=\alpha+1$, we finally get 
\[
|E(H)|\leq n + O\Big(n^{\frac{(\alpha+1)-(1+\epsilon)}{(\alpha+1)-1}}\Big)
= n + O\big(n^{\frac{\alpha-\epsilon}{\alpha}}\big)
= n + O\big(n^{1-\epsilon/\alpha}\big). 
\]
Observe that, by our choice of $d\sim (1-\epsilon/\alpha)\log_2n$ and $t=(1+\epsilon)\log_2n$, we get that $t-d$ grows to infinity when $n$ grows, which suffices for applying Lemma~\ref{lem:technical-asymp-1}. 

\paragraph{Fixing the number of nodes.}

Our construction guarantees that the graph $H$ and the parameter $d$ governing the size $2^{d-1}$ of the core $C$ satisfy 
\[
|V(C)|< n\leq |V(H)|.
\]
It follows that getting a graph $G$ from $H$ with the desired number of nodes $n=|V(G)|$ can be  done by merely removing nodes from the (truncated) Binomial trees $T_v$ pending from every node $v$ of the core, without having to modify the structure of the core itself. 
Removing nodes from a Binomial tree by iteratively removing leaves cannot increase the broadcast time of the root, nor it can increase the distance from a leaf to the root. As a consequence, the  $n$-node graph $G$ obtained from $H$ by removing nodes from the Binomial trees pending from the core satisfies
\[
b(G)\leq b(H)\leq (1+\epsilon)\log_2n, 
\]
and 
\[
|E(G)|\leq |E(H)|\leq n+O(n^{1-\epsilon/\alpha}),
\]
as desired. 
This completes the proof of Theorem~\ref{thm:upper-bound-general}. 
\end{proof}

Theorem~\ref{thm:upper-bound-general} has an important corollary. Recall that the \emph{broadcast cost}~\cite{Harutyunyan14} of a graph $G$ is defined as $\cost(G)=b(G)\cdot|E(G)|$. That is, every edge of $G$ is charged for $b(G)$ units of cost. For every $n\geq 1$, let $\cost(n)$ be the minimum, taken over all $n$-node graphs, of the broadcast costs of these graphs. We show that the upper bound in~\cite{Harutyunyan14} on this minimum cost is a direct consequence of Theorem~\ref{thm:upper-bound-general}.

\begin{corollary}\label{cor:min-cost}
    Let $\alpha=1/\log_2\phi-1$ where $\phi=(\sqrt{5}+1)/2$ is the golden ratio. For every $n\geq 1$,  $\cost(n)=n\log_2n + \alpha \;n\log_2\log_2n + O(n)$.
\end{corollary}

\begin{proof}
    By Theorem~\ref{thm:upper-bound-general}, for every $n\geq 1$, and  for every $\epsilon\in(0,\alpha)$, we have 
    \[
    \cost(n)= (1+\epsilon)\log_2n \cdot\big(n+O(n^{1-\epsilon/\alpha})\big). 
    \]
    Set $\epsilon=\frac{\alpha\;\log_2\log_2 n}{\log_2n}$. We then get that 
    \[
    \cost(n)=n\log_2n + \alpha \; n\log_2\log_2n + O\big(\log_2n\cdot n^{1-\frac{\log_2\log_2n}{\log_2n}}\big).
    \]
    Since $n^{-\frac{\log_2\log_2n}{\log_2n}}=2^{-\log_2\log_2n}=(\log_2n)^{-1}$, we eventually get that 
    \[
    \cost(n)=n\log_2n + \alpha \;n\log_2\log_2n + O(n), 
    \]
    as claimed. 
\end{proof}

\paragraph{Remark.} 

The $n$-node graph $G$ minimizing the $c(G)$ in the proof of Corollary~\ref{cor:min-cost} has broadcast time 
$
b(G)=\log_2n+O(\log\log n).
$
It may be the case that minimizing the cost could be obtained with a graph $G$ with broadcast time $b(G)=\log_2n+O(1)$. This would however hold only if such a graph exists, with $n+o(n/\log n)$ edges. In this case, this graph would have a cost equal to $n\log_2n+O(n)$, matching the known lower bound $n\log_2n+\Omega(n)$ from~\cite{Harutyunyan14}. Actually, we believe that such a graph $G$ does \emph{not} exist (see Section~\ref{sec:conclusion}). That is, we believe that every $n$-node graph with broadcast time $\log_2n+c$ for some constant $c\geq 0$ must have $n+\Omega(n)$ edges, where the constant hidden in the big-$\Omega$ notation depends on~$c$. We shall prove that this is the case for $c=1$ in Section~\ref{sec:lwb}.

\section{Lower Bound}
\label{sec:lwb}

In this section, we prove Theorem~\ref{thm:lower-bound}, that is, we show that, 
for every $n\geq 1$ power of~2, any $n$-node graph with broadcast time $\log_2n+1$ has at least $\frac98n$ edges. 
We start with establishing a result of independent interest about the structure of subtrees of Binomial trees. 

\begin{figure}[tb]
\centerline{\includegraphics[scale=0.5]{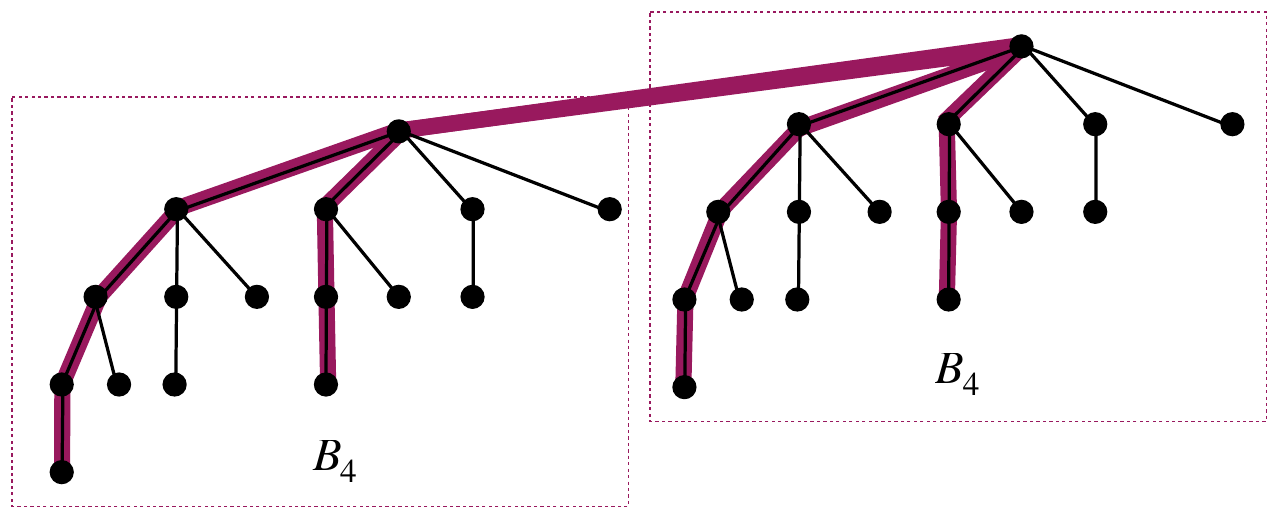}}
\caption{\sl The Binomial tree $B_4$ contains a 8-node subtree with two leaves, and the Binomial tree $B_5$ contains a 16-node subtree with four leaves. In general, for $k\geq 2$, the Binomial tree $B_k$ contains a $2^{k-1}$-node subtree with $2^{k-3}$ leaves, which is the least number of leaves for a $2^{k-1}$-node subtree. }
\label{fig-2-leaves}
\end{figure}

\begin{lemma}
\label{lem-number-leaves-subtree-BT}
Let $k\geq 1$ be an integer, and let $T$ be a subtree of the $k$-dimensional Binomial tree~$B_k$, rooted at the root of $B_k$, and with $n=2^{k-1}$ vertices. The number of leaves of $T$ is at least $n/4 = 2^{k-3}$.
\end{lemma}

\begin{proof}
   As a subtree of $B_k$, the tree $T$ can be obtained by removing from $B_k$ a collection $C$ of $m\geq 1$ subtrees of~$B_k$ where each tree $B_i\in C$, $1\leq i\leq m$, is itself isomorphic to a Binomial tree of dimension $d_i\in\{0,\dots,k-1\}$, with $x_i=2^{d_i}$ vertices. Since $T$ has $2^{k-1}$ vertices, we get that the number of removed vertices satisfies
    \begin{equation}
        \label{eq:number-of-nodes}
        \sum_{i=1}^m x_i = 2^{k-1}. 
    \end{equation}
    The Binomial tree $B_k$ has $2^{k-1}$ leaves. Let us count how many leaves are removed after the removal of all the Binomial trees in~$C$. Removing a tree $B_i$ of dimension~$d_i$>0 removes at most $2^{d_i-1}=x_i/2$ leaves from~$B_k$. Removing a tree $B_i$ of dimension~$d_i$=0 removes at most one leaf from~$B_k$. Let us enumerate the trees in $C$ as $B_1,\dots,B_{m'},B_{m'+1},\dots,B_m$ where the $m'\geq 0$ first trees satisfy $d_i>0$, and all the remaining subtrees satisfy $d_1=0$. We get that the total number of leaves removed from $B_k$ is at most 
    \[
    R=\Big (\frac12\sum_{i=1}^{m'} x_i\Big ) + (m-m').  
    \]
    We now aim  at determining how large $R$ can be  under the constraint of Eq~\eqref{eq:number-of-nodes}. 
    Since the number of leaves of $B_k$ is at most $2^{k-1}$, we get that $m-m'\leq 2^{k-1}$. However, half of these $2^{k-1}$ leaves of $B_k$ are connected to a vertex of degree~2 in $B_k$, which implies that removing such a leaf actually keeps the number of leaves unchanged. In other words, at most  $2^{k-2}$ leaves of $B_k$ may contribute to~$R$, and it can only be better to remove Binomial trees of dimension larger than~$0$ as far as maximizing~$R$ is concerned. That is, our goal is in fact to solve 
    \[
    \text{max}\;\; \frac12\sum_{i=1}^{m'} x_i \;\; 
    \text{under the constraint} \; \sum_{i=1}^{m'} x_i = 2^{k-2}.
    \]
    This maximum is therefore at most~$\frac{2^{k-2}}{2}=2^{k-3}$, from which it follows that 
    \[
    R\leq 2^{k-2}+2^{k-3}. 
    \]
    As a consequence, $T$ has at least $2^{k-1}-(2^{k-2}+2^{k-3})=2^{k-3}$ leaves. 
\end{proof}

Note that the lower bound in Lemma~\ref{lem-number-leaves-subtree-BT} is tight for all $k\geq 3$. In particular, there is a 8-node subtree of $B_4$ with 2 leaves --- see Figure~\ref{fig-2-leaves}. The construction on Figure~\ref{fig-2-leaves} illustrating the tightness of Lemma~\ref{lem-number-leaves-subtree-BT} for $k\in\{4,5\}$ generalizes to all $k\geq 3$. 

The proof of Theorem~\ref{thm:lower-bound} is based on analyzing the number of leaves in the broadcast tree of any root node~$v$, and Lemma~\ref{lem-number-leaves-subtree-BT} is used for deriving a lower bound on this number of leaves, which is in turn used for lower bounding the number of nodes with degree at most~2 in $n$-node graphs with broadcast time $\lceil\log_2n\rceil+1$. 

\begin{proof}[Proof of  Theorem~\ref{thm:lower-bound}]
Let $G=(V,E)$ be an $n$-node graph with $n=2^k$, and let us assume that the broadcast time $b(G)\leq k+1$. 
We consider two cases depending on whether the minimum degree of~$G$, denoted by $\delta(G)$, is at least~2 or not. 

Let us first assume that $\delta(G) \geq 2$, and let us consider a node $v$ with highest degree in~$G$, i.e., $\deg_G(v)=\max_{u\in V}\deg_G(u)$. Consider a $(k+1)$-round broadcast protocol from $v$ in~$G$. This protocol induces a subtree $T_v$ of~$G$, referred to as the broadcast tree of~$v$. Since $b(T_v,v)\leq k+1$, the tree $T_v$ is also a subtree of the $(k+1)$-dimensional Binomial tree $B_{k+1}$, where $v$ is the root of $B_{k+1}$.  
Since $\delta(G) \geq 2$, all the leaves of $T_v$ must have degree at least~2 in~$G$. Thus, every leaf in $T_v$ must be incident to another edge in graph~$G$. By Lemma~\ref{lem-number-leaves-subtree-BT}, $T_v$~has at least $n/4=2^{k-2}$ leaves. Even if all these leaves were connected to each other in~$G$, $G$~must still contain at least $\frac12 \cdot 2^{k-2}= 2^{k-3}$ additional edges. Thus, in total, $G$ must have a number of edges at least 
$
n - 1 + 2^{k-3} \geq  n - 1 + \frac{n}{8} = \frac{9}{8}n - 1.
$

Let us now assume that $\delta(G)=1$, and let $u$ be a node of degree~1 in~$G$. Again, let us consider a broadcast protocol from $u$ in $G$ performing in at most $k+1$ rounds, and let $T_u$ be the associated broadcast tree. Node $u$ has a unique neighbor $v$ in~$G$, and thus $u$ calls $v$ at the first round of the broadcast protocol. Node $v$ is the root of the tree $T_v = T_u\smallsetminus\{u\}$ containing $2^k-1$ vertices. In particular, $T_v$ is actually $B_k$ except one leaf vertex.  Broadcasting from $v$ in $T_v$ takes at most $k$ rounds, which implies that $T_v$ is a subtree of the $k$-dimensional Binomial tree~$B_k$. $T_v$ has $2^{k-1}-1$ leaves, and thus $2^{k-1}$ internal nodes. 
Assume that, out of the $2^{k-1}-1$ leaves of $T_v$, $x$~vertices have degree at least 2 in~$G$, and the remaining $(2^{k-1}-1) - x$ have degree~1 in~$G$. That is, including~$u$, $2^{k-1}- x$ nodes have degree~1 in~$G$. Even all the $x$ vertices of degree at least~2 were of degree exactly~2, and were connected to each other in graph~$G$, $G$~would still have at least $n - 1 + \lceil \frac{x}{2} \rceil$ edges.

If $x \geq 2^{k-2} = \frac{n}{4}$, then $G$ contains at least $\lceil \frac{x}{2}\rceil \geq \lceil \frac{n}{4} \rceil = \frac{n}{8}$ additional edges, which proves the theorem, and we are done.  
So, let us now assume that $x \leq 2^{k-2} - 1$. That is, the tree $T_v$ has at least 
\[
(2^{k-1}-1) - x \geq (2^{k-1}-1) - (2^{k-2}-1) \geq 2^{k-2}
\]
leaves with degree~1 in~$G$. Let $\ell$ be a leaf of $T_v$ with degree~1 in~$G$. 
To broadcast in $G$ in $k+1$ rounds from node~$\ell$, this node must inform its neighboring vertex~$w$ at round~1. To be able to broadcast from $w$ to the $2^k-1$ remaining vertices, in $k$ rounds in the graph $G\smallsetminus \{\ell\}$, it must be that $\deg_{G\smallsetminus \{\ell\}}(w) \geq k-1$. This lower bound must hold for each and every leaf $\ell$ of tree $T_v$ with degree~1 in~$G$. It follows that any internal vertex of $T_v$ that has a neighbor of degree~1 in~$G$ must have degree at least $k - 1$ in $G$.  

Every vertex of~$B_k$, as well as of~$T_v$, has at most one neighbor of degree 1. The total number of internal vertices of $T_v$ is at least the half of its total number of vertices. So, $T_v$ contains $\lceil \frac{2^k-1}{2} \rceil = 2^{k-1}$ internal vertices. As $x < 2^{k-2}$, we get that, out of these $2^{k-1}$ internal vertices,   $2^{k-1} - x > 2^{k-2}$ must have a neighbor of degree~1 in graph $G$. So, to count the number of additional edges in graph~$G$, we now focus on the degree distribution of the internal vertices of tree~$T_v$, which is the same as in the Binomial tree~$B_k$. 

Let $S_i$ be the set of vertices of $T_v$ with degree $i$. We have  $|S_k|=2$, and, for all $j\in\{2,\dots,k-2\}$, $|S_j| = 2^{k-j}$. This amounts for a total of 
\[
|S_k|+|S_{k-1}|+...+|S_2| = 2+2+2^2+...+2^{k-2} = 2^{k-1}
\]
internal vertices in $T_v$.  
For every  $i\in\{2,\dots,k\}$, let $x_i$ be the number of vertices of $S_i$ that have a neighbor of degree~1 in graph $G$. We have $0 \leq x_i \leq 2^{k-i}$, from which it follows that there are $x_i(k-1-i)$ additional edges going out of the $x_i$ vertices of~$S_i$. In total,  
\[
\sum_{i=2}^{k-2} x_i(k-1-i) 
\]
additional edges incident to all internal vertices. Even if all these vertices were connected to each other,  the number of additional edges in $G$ would still be at least 
$
\frac{1}{2} \sum_{i=2}^{k-2} x_i(k-1-i).
$
Note that the vertices of $S_k$ and $S_{k-1}$, which amount for four vertices in total,  have already degree at least $k-1$. Thus, we are let to 
\begin{align*}
& \text{minimize}\; A=  \frac{1}{2} \sum_{i=2}^{k-2} x_i(k-1-i) \\
& \mbox{under the constraints $\sum_{i=2}^{k-2} x_i\geq 2^{k-2} - 3$, where $0 \leq x_i \leq 2^{k-i}$ for every $i\in\{2,\dots,k\}$}. 
\end{align*}
We have $x_2 \geq 2^{k-2}-3 - \sum_{i=3}^{k-2} x_{i}$. Plugging this into~$A$, we get
\begin{align*}
A 
& = \frac{1}{2}[x_2(k-3) + \sum_{i=3}^{k-2} x_i(k-1-i)] \\
& \geq \frac{1}{2}[(k-3)(2^{k-2} - 3) - \sum_{i=3}^{k-2}x_i) + \sum_{i=3}^{k-2} x_i(k-1-i)] \\
& = \frac{1}{2}(k-3)(2^{k-2} - 3) - \frac{1}{2}\sum_{i=3}^{k-2}[ x_i(k-3) - x_i(k-1-i)] \\
& = \frac{1}{2}(k-3)(2^{k-2} - 3) - \frac{1}{2}\sum_{i=3}^{k-2} x_i(i-2)\\
& \geq \frac{1}{2}(k-3)(2^{k-2} - 3) - \frac{1}{2}\sum_{i=3}^{k-2} 2^{k-i}(i-2),
\end{align*}
where the last inequality holds because $x_i \leq 2^{k-i}$. 
Since all the terms of the last summation are upper bounded by $2^{k-3}$ as $2^{k-i}(i-2) \leq 2^{k-3}$ for every $i \geq 3$, we get 
\begin{align*}
A 
& \geq \frac{1}{2}(k-3)(2^{k-2} - 3) - \frac{1}{2}\sum_{i=3}^{k-2} 2^{k-3} \\
& = \frac{1}{2}(k-3)(2^{k-2} - 3) - \frac{1}{2}(k-4)2^{k-3}\\
& = \frac{1}{2}[(k-3)2^{k-2} - 3(k-3) - (k-4)2^{k-3}] \\
& = \frac{1}{2}2^{k-3}(2k-6-k+4) - \frac{3}{2}(k-3) \\
& = 2^{k-4}(k-2) - \frac{3}{2}(k-3). 
\end{align*}
Therefore, in total, $G$~must contain at least 
$
2^{k-4}(k-2) - \frac{3}{2}(k-3)
$
additional edges. Finally, we get that  
$
|E(G)| \geq n - 1 + 2^{k - 4}(k - 2) - \frac{3}{2}(k-3) = n - 1 + (\log_2 n - 2) \cdot \frac{n}{16} - \frac{3}{2}(\log_2 n - 3) \geq n-1+\frac{n}{8}
$
for $k \geq 6$, i.e., $n \geq 64$. This concludes the proof.
\end{proof}

\section{1-Relaxed Minimum Broadcast Graphs}
\label{sec:1-relaxed}

This section is establishing Theorem~\ref{thm:upper-bound-tau-1}, which says that, for every $n\geq 1$, there exists an $n$-node graph with broadcast time $\lceil\log_2n\rceil+1$ and at most $2n-4\log_2n+O(1)$ edges. 

\begin{proof}[Proof of Theorem~\ref{thm:upper-bound-tau-1}]
Our construction is a refinement of the construction in~\cite{GrigniP91}, which was used to show that for every $n\geq 1$, there exists an $n$-node graph $G$ with broadcast time $\lceil\log_2n\rceil+1$ and at most $2n-\lceil\log_2n\rceil-2$ edges. Let $k=\lceil\log_2n\rceil$. The graph $G$ constructed in~\cite{GrigniP91} is obtained from the Binomial tree $B_k$. \begin{itemize}
    \item First, $B_k$ is pruned so that the resulting rooted tree $T$ has $n$ nodes. This pruning is performed by removing nodes from $B_k$ bottom-up, i.e., removing first the node at distance~$k$ from the root $r$ of~$B_k$, then the $\binom{k}{k-1}$ nodes at distance $k-1$, etc. Note that since $2^{k-1}<n\leq 2^k$, no children of the root of $B_k$ are removed by this pruning process. 
    \item Second, for every node $u$ different from the root $r$ of $B_k$, and different from any child of $r$ in $B_k$, an edge $\{u,r\}$ is added to $B_k$. 
\end{itemize}
We have $b(G)\leq k+1$. Indeed, for every source node $u$, the broadcast protocol consists of
\begin{enumerate}
    \item $u$ sends the message to $r$ (unless $u=r$, in which case this first round is skipped); 
    \item the $k$-round broadcast protocol from $r$ is performed in the subtree $T$ of $B_k$. 
\end{enumerate}
By construction, $|E(G)|\leq |E(T)|+(n-(k+1))=2n-k-2$. We show now how to save an additional $3k$ additive factor. 

A first observation in that one can use the symmetry of $B_k$ to save $k$ edges. Indeed, $B_k$ is obtained by connecting the two roots $r$ and $r'$ of two copies of $B_{k-1}$ by an edge (see Figure~\ref{fig:improvingGP}). Let $r$ be the root of $B_k$. Not only $b(B_k,r)=k$, but also $b(B_k,r')=k$. It follows that there is no need to connect each children of $r'$ to $r$. This saves $k-1$ edges whenever $k\geq 4$. Saving another $2k$ edges requires more work. 

\begin{figure}[htb]
\centerline{\includegraphics[scale=0.5]{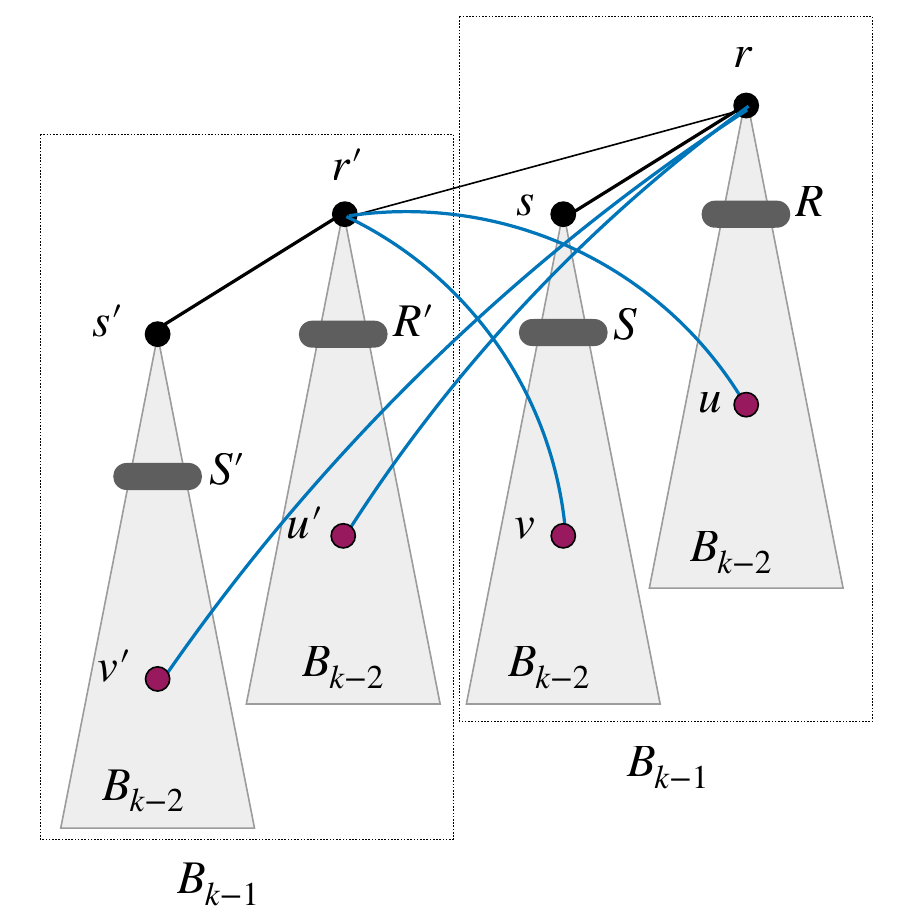}}
\caption{\sl  Illustration of the construction in the proof of Theorem~\ref{thm:upper-bound-tau-1}. The Binomial tree $B_k$ is composed of two Binomial trees $B_{k-1}$, each one composed of two trees $B_{k-2}$. The sets $R,S,R',S'$ are the respective children of the four roots $r,s,r',s'$ of these four trees. }
\label{fig:improvingGP}
\end{figure}

The Binomial tree $B_k$ actually consists of four copies of $B_{k-2}$, as illustrated on Figure~\ref{fig:improvingGP} where the two copies of the Binomial trees $B_{k-1}$ are each decomposed into two copies of $B_{k-2}$. In this figure, the sets $R,S,R'$, and $S'$ are the respective children of $r,s,r'$, and $r$ in the four $(k-2)$-dimensional Binomial trees. We claim that there is no need to connect the nodes in $S$ to $r$ (resp., the nodes in $S'$ to $r'$), and the corresponding edges can be removed from the previous construction, saving $|S|+|S'|=2(k-2)$ edges. However, this saving requires to rewire the other edges for preserving a $(k+1)$-round broadcast time. 

Specifically, instead of connecting every node $x$ in any of the two copies of $B_{k-1}$ to the root of its own $(k-1)$-dimensional Binomial tree, we ``cross'' the connections, i.e., we connect $x$ to the root of the other $(k-1)$-dimensional Binomial tree (see Figure~\ref{fig:improvingGP}). For $k\geq 6$, the total number of edges of the resulting graph $G$ is 
\[
|E(G)| = 2(n - 1) - 2\big(1+(k - 1) + (k - 2)\big)
= 2n - 4k + 2, 
\]
as claimed. It remains to check that $b(G)\leq k+1$. 

Let $x$ be a vertex of $G$. Since $b(B_k,r)=b(B_k,r')=k$, we have $b(G,x)\leq k+1$ for every node $x$ in the neighborhood of these two nodes in~$G$, that is, for every $x\notin S \cup S'$. So, let us consider a node $x\in S$ (the case $x\in S'$ is the same, by symmetry).   

Observe that, for $k\geq 8$, the bottom-up pruning of $B_k$ that was performed to get a tree $T$ with exactly $n$ nodes, $2^{k-1}< n \leq 2^k$, preserves the fact that $x$ has at least one child~$y$ in~$T$. Broadcast from $x$ proceeds as follows. 
\begin{itemize}
    \item Node $x$ calls $y$ at the first round. 
    \item At the second rounds, $y$ calls $r'$, and $x$ calls $s$. 
    \item At the third round, $r'$ calls $s'$, and  $s$ calls $r$. So, at the end of round~3, the roots of the four copies of the $(k-2)$-dimensional Binomial trees have received the message. \item An additional $k-2$ rounds suffices to broadcast in parallel in the four $(k-2)$-dimensional Binomial trees.
\end{itemize}
Thus $b(G,x)\leq 3+(k-2) = k+1$, which completes the proof. 
\end{proof}

\section{Conclusion}
\label{sec:conclusion}

This paper focuses on the minimum number of edges $(n-1)+h(n,\tau)$ of any $n$-node graph with broadcast time $\lceil\log_2n\rceil+\tau$. Combining our results with the ones in \cite{GrigniP91,KhachatrianH89,Labahn89}, the overhead function $h(n,\tau)$ satisfies, for every $n\geq 1$,  $h(n,0)=\Theta(n\cdot L(n))$, and 
\begin{equation}\label{eq:B1}
    \frac{n}{8}\leq h(n,1)\leq n-4\log_2n+O(1).
\end{equation}
Moreover, for every $\epsilon>0$,
\[
h(n,\epsilon\cdot\log_2n)=
\left\{\begin{array}{ll}
O(n^{1-\epsilon/\alpha}) & \mbox{if $\epsilon<\alpha=1/\log_2\phi-1$ where $\phi=(\sqrt{5}+1)/2$}\\
0 & \mbox{otherwise.}
\end{array}\right.
\]
We refer to Figure~\ref{fig:summary} for an illustration of these results. 

Equation~\eqref{eq:B1} shows that there is still a large gap between the best known upper and lower bounds on the smallest number of edges of graphs with broadcast time $\lceil\log_2n\rceil+1$. Closing this gap is an interesting open problem. 

For the general case, we actually believe that our upper bound on $h(n,\epsilon\cdot\log_2n)$ is tight, at least as far as the degree of the polynomial in~$n$ is concerned. This belief is supported by the fact that the parameters in our construction are optimized for producing the sparsest graphs. Since the construction in~\cite{AverbuchPR17} produces denser graphs, it appears that beating the $O(n^{1-\epsilon/\alpha})$ bound should requires to come up with a radically different construction, where the graph is not composed of a relatively dense core to which trees are attached, but of a tree to which are added $o(n^{1-\epsilon/\alpha})$ edges in a specific manner guaranteeing broadcast time at most $(1+\epsilon)\log_2 n$, in a way similar to the proof of Theorem~\ref{thm:upper-bound-tau-1} in which edges are added to a subtree of a Binomial tree. Such a construction may exist, but its design seems challenging. 

\begin{conjecture}
    Let $\alpha=1/\log_2\phi-1$ where $\phi=(\sqrt{5}+1)/2$ is the golden ratio. For every $\epsilon\in(0,\alpha)$,  any $n$-node graph with broadcast time at most $(1+\epsilon)\log_2n$ has at least $n+\tilde\Omega(n^{1-\epsilon/\alpha})$ edges, where the $\tilde\Omega$-notation hides possible polylogarithmic factors in~$n$. 
\end{conjecture}

A case of particular interest is $n$-node graphs with broadcast time $\lceil\log_2n\rceil+c$ where $c$ is a constant. 

\begin{conjecture}
    For every integer $c\geq 1$, any $n$-node graph with broadcast time at most $\lceil\log_2n\rceil+c$ has at least $n+\Omega(n)$ edges. 
\end{conjecture}

Thanks to Theorem~\ref{thm:lower-bound}, we know that this conjecture holds for $c=1$, at least for infinitely many values of~$n$. We believe that the conjecture holds for all~$n$, and every constant $c\geq 1$.

\newpage
\bibliographystyle{plainurl}
\bibliography{references}

\end{document}